\documentclass[submission,copyright,creativecommons]{eptcs}
\providecommand{\event}{EXPRESS 2011}

\usepackage[utf8]{inputenc}
\usepackage[T1]{fontenc}
\usepackage{amsfonts}
\usepackage{amssymb}
\usepackage{graphics}
\usepackage{latexsym}
\usepackage{url}
\usepackage{pst-all}
\usepackage[all]{xy}
\usepackage{wasysym} 

\newcommand{\Lo}{{\cal L}}

\newcommand{\ioi}{\Leftrightarrow}

\newcommand{\der}[1]{\ensuremath{\stackrel{#1}{\longrightarrow}}}

\newcommand{\lt}{\sqsubseteq}

\newcommand{\Proc}{{\bf P}}

\newcommand\lb {[\![}
\newcommand\rb{]\!]}

\newcommand{\putaway}[1]{}
\newcommand{\must}[1]{[ #1 ]}
\newcommand{\may}[1]{\langle #1 \rangle}

\newcommand{\sem}[1]{\relax\ifmmode \lb #1 \rb \else $\lb #1 \rb$ \fi}

\newcommand{\calC}{\mathcal{C}}
\newcommand{\calM}{\mathcal{M}}
\newcommand{\com}{\mathop{\circ}}

\newcommand{\ccsim}{\mathbin{\lesssim_{cc}}}
\newcommand{\cclt}{\ccsim}
\newcommand{\tbi}[1]{\mathcal{T}(#1)}

\newcommand{\tabi}[1]{\mathcal{T}^{+}_{c^l}(#1)}

\newcommand{\qed}{~\Square}

\newtheorem{definition}{Definition}
\newtheorem{theorem}{Theorem}
\newtheorem{proposition}{Proposition}
\newtheorem{corollary}{Corollary}
\newtheorem{lemma}{Lemma}

\newtheorem{remark}{Remark}

\newenvironment{discussion}{\par\addvspace{\bigskipamount}
\noindent{\textbf{Discussion}}\ }{\par\addvspace{\bigskipamount}}

\newenvironment{proof}{\par\addvspace{\bigskipamount}
\noindent\textit{\textbf{Proof.}}\ }{\par\addvspace{\bigskipamount}}

\title{Graphical representation of covariant-contravariant modal
formulae} 
\author{Luca Aceto\quad Anna Ing\'olfsd\'ottir
\institute{ICE-TCS, School of Computer Science\\ Reykjavik
University\thanks{Research supported by the project `Processes and
Modal Logics' (project nr.~100048021) of the Icelandic Research Fund,
and the Abel Extraordinary Chair programme within the NILS Mobility
Project.}\\ Iceland} \and Ignacio F\'abregas \quad David de Frutos
Escrig\quad Miguel Palomino \institute{Departamento de Sistemas
Inform\'aticos y Computaci\'on\\ Universidad Complutense de
Madrid\thanks{Research supported by Spanish projects DESAFIOS10
TIN2009-14599-C03-01, TESIS TIN2009-14321-C02-01 and PROMETIDOS
S2009/TIC-1465}\\ Spain} } 

\def\titlerunning{Graphical representation of cc-modal formulae} 
\def\authorrunning{L. Aceto, I. F\'abregas, D. de Frutos,\\\, A. Ing\'olfsd\'ottir \& M. Palomino}

\begin{document}
\maketitle

\begin{abstract}
Covariant-contravariant simulation is a combination of standard
(covariant) simulation, its contravariant counterpart and
bisimulation. We have previously studied its logical characterization
by means of the covariant-contravariant modal logic. Moreover, we have
investigated the relationships between this model and that of modal
transition systems, where two kinds of transitions (the so-called may
and must transitions) were combined in order to obtain a simple
framework to express a notion of refinement over state-transition
models.  In a classic paper, Boudol and Larsen established a precise
connection between the \textit{graphical} approach, by means of modal
transition systems, and the \textit{logical} approach, based on Hennessy-Milner
logic without negation, to system specification. They obtained a (graphical) representation theorem proving that a formula can be represented by a term if, and only if, it is consistent and prime. We show in this paper that the formulae from the covariant-contravariant modal logic that admit a ``graphical'' representation by means of processes, modulo the covariant-contravariant simulation preorder, are also the consistent and prime ones.
In order to obtain the desired graphical representation result, we first restrict ourselves
to the case of covariant-contravariant systems without bivariant
actions. Bivariant actions can be incorporated later by means of an
encoding that splits each bivariant action into its covariant and its
contravariant parts.
\end{abstract}

\section{Introduction}

{\em Modal transition systems} (MTSs) were introduced
in~\cite{Larsen89,LarsenT88} as a model of reactive computation based
on states and transitions that naturally supports a notion of {\em
refinement}. This is connected with the use of Hennessy-Milner Logic
without negation as a specification language: a specification
describes the collection of (good) properties that any implementation
has to fulfil. More generally, a process $p$ is considered to be
better than $q$ if the set of formulae satisfied by $q$
is included in the set of formulae satisfied by $p$. The tight connections
between these two ways of expressing the notions of specification and
refinement were studied in~\cite{BoudolL1992}. There the authors
talked about ``graphical'' representation (by means of one or several
MTSs) of logical specifications, and completely characterized the
collection of logical specification that can be ``graphically
represented''. These are the so-called prime, consistent formulae.

There are two types of modal operators in Hennessy-Milner Logic: $\langle
a\rangle$ and $[a]$, for each action $a$. Intuitively, a formula
$\langle a\rangle\varphi$ indicates that it must be possible to execute
$a$ and reach a state that satisfies $\varphi$, while $[a]\varphi$ imposes
that this will happen after any execution of $a$ from the current
state. It is well known that these two operators reflect the duality
$\exists$-$\forall$, so that any process satisfying a $\langle
a\rangle\varphi$ formula {\em must} include some $a$-labelled
transition reaching a state satisfying $\varphi$, whereas the
constraint expressed by a $[a]\varphi$ formula is better understood in
a negative way: a process satisfying it {\em may not} contain an
$a$-labelled transition reaching a state that does not satisfy $\varphi$. In
particular, the formula $[a]\bot$ indicates that a process cannot
execute $a$ in its initial state, and therefore, using these formulae,
we can limit the set of actions offered at any state. 

In order to reflect these two kinds of constraints at the
``operational'' level, MTSs contain two kinds of transitions: the {\em
may} transitions and the {\em must} transitions. Then we can use MTSs
both as specifications or as implementations, and the notion of
refinement imposes that, in order to implement correctly a
specification, an implementation should exhibit all the \textit{must}
transitions in the MTS that describes the specification and may not
include any transition that is not allowed by the specification: we
cannot add any new \textit{may} transition, although those in the
specification could either disappear, be preserved or turned into
\textit{must} transitions. The relation between \textit{may} and
\textit{must} is reflected in the formal definition of MTSs by
requiring that each must transition is also a may transition.

The conditions defining the notion of refinement between MTSs
obviously resemble those defining simulation and bisimulation. For may
transitions we have a contravariant simulation condition, expressing
the fact that no new (non-allowed) \textit{may} transition can appear
when refining a specification. Since we impose that \textit{must}
transitions induce the corresponding \textit{may} transitions, we
could think that they are related in a ``bisimulation-like'' style.
However, this is not the case since the contravariant simulation
condition imposed on the may part can be covered by a \textit{may}
transition without \textit{must} counterpart. In fact, this is crucial
in order to capture the principle that a \textit{may} transition can
be refined by a \textit{must} transition.

Some of the authors of this paper thought that a more direct
combination of simulation and bisimulation conditions could capture in
a more flexible way all the ideas on which the specification of
systems by means of modal systems and modal logics is based, and we
looked for the clearest and most general framework to express those
modal constraints. We found that covariant-contravariant systems
(sometimes abbreviated to cc-systems) are a possible answer to this
quest, combining pure (covariant) simulation, its contravariant
counterpart and bisimulation.

We started the study of {\em covariant-contravariant simulation}
in~\cite{FabregasFP09}, and the modal logic characterizing it was
presented in~\cite{FabregasEtAl10-logics}. (In what follows, we refer
to this logic as cc-modal logic.) In the most general case, we
consider a partition of the set of actions into three sets: the
collection of covariant actions, that of contravariant actions, and
the set of bivariant actions. Intuitively, one may think of the
covariant actions as being under the control of the specification LTS,
and transitions with such actions as their label should be simulated
by any correct implementation of the specification. On the other hand,
the contravariant actions may be considered as being under the control
of the implementation (or of the environment) and transitions with
such actions as their label should be simulated by the
specification. The bivariant actions are treated as in the classic
notion of bisimulation.

We will see in this paper that, as in the MTS setting, the consistent and prime formulae from the cc-modal logic are exactly those that admit a ``graphical'' representation by means of processes modulo the covariant-contravariant simulation preorder. Moreover, each formula in the cc-modal logic can be represented ``graphically'' by a (possibly empty) finite set of
processes.

The proofs of these representation results are inspired by the developments in~~\cite{BoudolL1992}. There are, however, subtle differences
because, in covariant-contravariant systems, each action has a single
modality (covariant, contravariant, bivariant), while in MTSs we can
combine both \textit{may} and \textit{must} transitions.

In fact, in order to obtain the desired graphical representation, for
technical reasons we first restrict ourselves to the case of
covariant-contravariant systems without bivariant actions. The reason
that justifies this constraint is that bivariant actions cannot be
approximated in a non-trivial way (either we have one of them as
itself, or we do not have it at all). Instead, covariant and
contravariant actions behave in a more flexible way and we can obtain
the desired characterization result by following the lead of the work
done for MTSs.

Then we observe that bivariant actions can be seen as the combination
of a covariant and a contravariant action. In fact, this also
corresponds with the idea used in~\cite{AcetoEtAl11} when relating
MTSs and cc-systems. Indeed, the constraint imposed on \textit{must}
transitions in MTSs, where they should always be accompanied by their
\textit{may} counterparts, tells us somehow that they have a
``nearly'' bivariant behaviour. (To be more precise, they are first
covariant, but they are also ``semi''-contravariant because when
comparing two processes $p$ and $q$, any \textit{must} transition in
$q$ should fit with either a corresponding \textit{must} transition in
$p$, or at least with a \textit{may} transition there.)

We could say that the very recent development of the notion of {\em
  partial bisimulation} in the setting of labelled transition systems
  (LTSs) presented in~\cite{Baetenetal} has completed the spectrum of
  modal simulations. Partial bisimulation combines plain
  bisimulation~\cite{Milner89,Park81} and simulation, also by means of
  a partition of the set of actions. For the actions in the
  distinguished set $B$ we have bisimulation-like conditions, while
  for the others we only impose simulation. Note that, instead,
  \textit{may} transitions in MTSs corresponded to contravariant
  simulation conditions, and therefore, partial bisimulation can be
  seen as a dual of MTSs, and covariant-contravariant systems
  (cc-systems) as a unifying framework where we can combine the
  refinement ideas in the theory of MTSs with the explicit
  consideration of the constraints imposed by the environment, which
  is possible when partial bisimulation is used. Once we know that the
  formulae from the modal logic for cc-systems also afford a graphical
  representation, we will be able to integrate the logical formulae
  into the development of systems using any of the models discussed
  above.

The remainder of the paper is organized as
follows. Section~\ref{background} is devoted to the necessary
background on covariant-contravariant simulations, whereas in
Section~\ref{sec:logic} we summarize the results on
covariant-contravariant modal formulae. In Section~\ref{sec:graph} we
develop the study of the graphical representation of cc-modal formulae
for processes without bivariant actions. Afterwards, in
Section~\ref{sec:bivar}, we show how we can work with cc-systems with
bivariant actions. Finally, Section~\ref{sec:future} concludes the
paper and describes some future research that we plan to pursue.

\section{Covariant-contravariant systems}\label{background}

We start the technical part of the paper by defining the
covariant-contravariant simulation semantics for processes. Our
semantics is defined over {\sl Labelled Transition Systems} (LTS)
$S=(\Proc,A,\der{})$, where $\Proc$ is a set of process states, $A$ is
a set of actions and $\der{}\subseteq\Proc\times A\times\Proc$ is a
transition relation on processes. We follow the standard practice and
write $p\der{a}q$ instead of $(p,a,q)\in\der{}$.  Because of the
covariant-contravariant view, we assume that $A$ is partitioned into
$A^l$ and $A^r$, expressed as $A=A^l\uplus A^r$. As we have already
mentioned in the introduction, we will delay the consideration of the
general case where we have also bivariant actions in a third class
$A^\mathit{bi}$ until Section~\ref{sec:bivar}.

Covariant-contravariant simulation can now be defined as follows:

\begin{definition}\label{Def:CCsim}
  Let $S=(\Proc,A^l\uplus A^r,\der{})$ be an LTS. A {\em
    covariant-contravariant simulation} over $S$ is a relation
  ${\mathrel{R}}\subseteq{\Proc\times\Proc}$ such that, whenever $p,q\in\Proc$ and
  $p \mathrel{R} q$, we have:
\begin{itemize}
\item For all $a\in A^r$ and all $p\der{a} p'$, there
 exists some $q\der{a}q'$ with $p' \mathrel{R} q'$.
\item For all $a\in A^l$ and all $q\der{a}q'$, there
 exists some $p\der{a}p'$ with $p' \mathrel{R} q'$.
\end{itemize}
We will write $p\ccsim q$ if there exists a covariant-contravariant simulation
$R$ such that $p \mathrel{R} q$.
\end{definition}

\begin{remark}
\emph{Note that we call the actions in $A^r$ like that, because for those there is a ``plain simulation'' from left to right; whereas for the actions in $A^l$ there is an ``anti-simulation'' from right to left.}
\end{remark}

It is well known that the relation $\ccsim$ is a preorder.

In this study we will be mainly concerned with ``finite'' properties of
systems, which will be either captured by (finite) logic formulae, or by
finite processes that can be described by means of process terms.

\begin{definition}
  Assume that $A=A^l\uplus A^r$. Then the collection of {\em process terms}, ranged over by
  $p,q$ etc. is given by the following syntax:
\[
p::=0\mid \omega\mid a.p\mid p+p,
\]
where $a\in A$. We denote  the set of process terms by $\mathcal{P}$.

The size of a process term is its length in symbols.
\end{definition}

We note that our set $\mathcal{P}$ of process terms is basically the
set of $BCCSP$ terms introduced in~\cite{VanGlabbeek01}. The only
addition to the signature of BCCSP is the constant $\omega$, which
will be used to denote the least LTS modulo $\ccsim$. However, we
assume a classification of the actions in two (disjoint) sets,
although this is not reflected in the syntactic structure of the
terms. Even if $\mathcal{P}$ only contains finite terms, by means of
$\omega$ we will obtain the full contravariant process which can
execute any action at any time.

In \cite{FabregasFP09,FabregasEtAl10-sos,FabregasEtAl10-logics} we
used a more general definition for covariant-contravariant simulations
which includes also bivariant actions, but since in the presence of
these bivariant actions some technical problems appear (in particular
the process $\omega$ will not be the least process with respect to the
covariant-contravariant simulation preorder), we have preferred to
first develop all the results without bivariant actions and, in
Section~\ref{sec:bivar}, we will describe how they can be extended to
a setting with bivariant actions.

\begin{definition} The {\em operational semantics} of $\mathcal{P}$ is
  defined by the following rules:
\begin{itemize}
\item $\omega\der{b}\omega$ for all $b\in A^l$,

\item $a.p\der{a}p$ for all $a\in A$,

\item $p\der{a}p'$ implies $p+q\der{a}p'$,

\item $q\der{a}q'$ implies $p+q\der{a}q'$.
\end{itemize}
\end{definition}
Observe that if $p\neq \omega$ and $p\der{a}p'$, then the size of $p'$
is smaller than the size of $p$.

It is clear that $\omega$ is the least possible element with respect
to the cc-simulation preorder. That is, we have $\omega\ccsim p$ for
any $p$. 

In what follows we assume that $A$ is finite.

\section{The covariant-contravariant modal logic}\label{sec:logic}

Covariant-contravariant modal logic has been introduced and studied
in~\cite{FabregasEtAl10-logics}.

\begin{definition}\label{Def:formulaeCC}
{\em Covariant-contravariant modal logic} $\Lo$ has the following syntax:
\[
\varphi ::= \bot \mid \top \mid \varphi\land\varphi \mid \varphi\lor\varphi\mid
            [b]\varphi \mid \langle a\rangle\varphi \qquad
            (a\in A^r, b\in A^l).
\]
The operators $\bot$, $\top$, $\land$ and $\lor$ have the standard
meaning whereas the semantics for the modal operators is defined as
follows:
\begin{itemize}
\item [] $p\models [b]\varphi$ if $p'\models \varphi$ for all
 $p\der{b} p'$,
\item [] $ p\models \langle a\rangle\varphi$ if $p'\models \varphi$ for
 some $p\der{a} p'$.
\end{itemize}
We say that a formula $\varphi$ is {\sl consistent}  if there is some $p$ such that
$p\models\varphi$.

The {\em modal depth} of a formula is the maximum nesting of modal
operators in it.
\end{definition}

The covariant-contravariant logic characterizes the
covariant-contravariant simulation semantics over image-finite
processes. Before we state this result formally we introduce some
notation.  We define the set of formulae that a process $p$ satisfies
by $\Lo(p)=\{\phi\mid p\models\phi\}$ and the logical preorder
$\lt_{\Lo}$ as follows: $p\lt_{\Lo}q$ iff
$\Lo(p)\subseteq\Lo(q)$. Recall that an LTS is {\em image finite} iff
the set $\{p' \mid p\der{a} p' \}$ is finite for each process $p$ and
action $a$.

Now we have the following theorem:

\begin{theorem}[\cite{FabregasEtAl10-logics}]\label{thm:log}
  If the LTS $S$ is image finite then $\ccsim=\lt_{\Lo}$ over $S$.
\end{theorem}

Clearly the processes in $\mathcal{P}$ are image finite.

\section{Graphical representation of formulae}\label{sec:graph}

Whenever we have a (modal) logic characterizing some semantics for
processes, we could look for a single formula that characterizes
completely the behaviour of a process logically; this is a so-called
{\sl characteristic formula}. This subject has been studied by many
authors in the literature, but we will just refer here to the
book~\cite{AcetoEtAl07b} for more details and further references to
the original literature.

It is clear that, since we only allow for finite formulae without any
fixed-point operator, we can only treat ``finite'' processes, such as
those definable by our simple process algebra $\mathcal{P}$. However,
the recursive definition of the characteristic formulae in what
follows gives us immediately the framework for extending our results
to finite-state processes following standard lines.

\begin{definition}\label{def:char-form}
  A formula $\phi\in\Lo$ is a {\em characteristic formula} for a process $p$
  iff $p\models\phi$ and $\forall q. (q\models \phi\Rightarrow p\cclt
  q)$.
\end{definition}

In what follows, we write $\phi\leq\psi$ if $\{p\in P\mid
p\models\phi\}\subseteq \{p\in P\mid p\models \psi\}$. We say that
$\phi$ and $\psi$ are logically equivalent, written $\phi\equiv\psi$,
iff $\phi\leq\psi$ and $\psi\leq\phi$.

\begin{lemma}\label{lem:char-form}
The following statements hold.
\begin{enumerate}
\item A formula $\phi\in\Lo$ is a characteristic formula for a process $p$ iff\, $\forall q.(q\models \phi\ioi p\cclt q)$. \item Assume that
$\chi(p)$ and $\chi(q)$ are characteristic formulae for processes $p$ and $q$, respectively. Then, we have that
\[
p\cclt q\mbox{ iff }\chi(q)\leq\chi(p).\]
\item A characteristic formula for a process $p$ is unique up to
logical equivalence.
\end{enumerate}
\end{lemma}

\begin{proof}
\begin{enumerate}
\item First assume that $\phi$ is a characteristic formula for a process
  $p$.  By definition $\forall q. ( q\models \phi\Rightarrow p\cclt
  q)$ holds. We have to prove that  $\forall q. (p\cclt q\Rightarrow
  q\models \phi)$. To this end, assume that $p\cclt q$.  As
  $p\models\phi$, by Theorem~\ref{thm:log} we have that $q\models
  \phi$ and we are done.

For the converse, as  $p\cclt p$ we have  that $p\models\phi$ and the
result follows.

\item Assume that $\chi(p)$ and $\chi(q)$ are characteristic formulae
for processes $p$ and $q$, respectively. First assume that $p\cclt q$
and that $r\models\chi(q)$. By Definition~\ref{def:char-form}, $q\cclt
r$ and thus $p\cclt r$. By the previous clause of the Lemma, also
$r\models\chi(p)$. As $r$ was arbitrary, this shows that
$\chi(q)\leq\chi(p)$. Next, assume that
$\chi(q)\leq\chi(p)$. As $q\models\chi(q)$ then
$q\models\chi(p)$, and by definition of the characteristic formula,
$p\cclt q$.

\item This claim follows directly from statement 2 above.\qed
\end{enumerate}
\end{proof}

As a characteristic formula for a process $p$ is unique up to logical
equivalence, we can denote it by $\chi(p)$ unambiguously. The next
lemma tells us that $\chi(p)$ exists for each process $p\in
\mathcal{P}$.

\begin{lemma}
  The characteristic formula for a process $p\in \mathcal{P}$ can be
  obtained recursively as
\begin{eqnarray*}
\chi(p) & = & \bigwedge_{p\der{a}p', a\in A^r}\may{a} \chi(p')\land
                \bigwedge_{b\in A^l}
        \must{b}(\bigvee_{p\der{b}p'}\chi(p'))\;\textrm{, if $p\neq\omega$.} \\
\chi(\omega) & = & \top.
\end{eqnarray*}
\end{lemma}

\begin{proof}
  First we prove that $p\models\chi(p)$, for each $p$. This follows by
  a simple induction on the size of $p$.

  Next we prove that, for any $q$, $q\models \chi(p)$ implies $p\cclt
  q$ by induction on the size of  $q$.

 First we note that if $p=\omega$ then $\chi(\omega)=\top$ and
  $\omega\cclt q$; hence we obtain the result. Also, for the case
  $p=0$, we have that $\chi(0)$ is equivalent to $\bigwedge_{b\in
  A^l}\must{b}\bot$. Thus if $q\models\chi(0)$, then the process $q$
  cannot perform any $b\in A^l$. This yields that $0\cclt q$.

  Now, let $p$ be a process different from $0$ and $\omega$, and
  assume that $q\models\chi(p)$. First suppose that $p\der{a}p'$ for
  some $p'$ and some $a\in A^r$. As $q\models\bigwedge_{p\der{a}p',
  a\in A^r}\may{a} \chi(p')$, this implies that there is some
  $q\der{a}q'$ with $q'\models\chi(p')$. Then, by induction, $p'\cclt
  q'$.

Next, assume that $q\der{b}q'$, for some $q'$ and $b\in A^l$. As
  $q\models\bigwedge_{b\in
  A^l}\must{b}(\bigvee_{p\der{b}p'}\chi(p'))$, we can conclude that
  $q'\models\chi(p')$, for some $p'$ with $p\der{b}p'$. Again, by induction, we conclude
  $p'\cclt q'$.\qed

\end{proof}

Next we consider the converse problem, we want to represent a formula
by a process, or at least by a finite set of
processes.
\begin{definition}\label{def:rep}
A formula $\phi$ is {\em represented by a (single) process} $p$ if
\[
\forall q\in \mathcal{P}.\;[q\models\phi\mbox{ iff }p\cclt q].
\]
A formula $\phi$ is {\em represented by a finite set} $M\subseteq
\mathcal{P}$ of processes if
\[
\forall q\in \mathcal{P}.\;[q\models\phi\mbox{ iff }\exists p\in M.\;p\cclt q].
\]
\end{definition}

It is clear that $p$ represents $\phi$ iff $\{p\}$ represents
$\phi$. Moreover, the empty set of processes represents the formula
$\bot$.

The following lemma connects the notion of ``graphical
representation'' of formulae with that of characteristic formula for
processes.

\begin{lemma}\label{lem:char1} We have the following properties:
\begin{enumerate}
\item $p$ represents  $\phi$ iff $\phi\equiv\chi(p)$.

\item \label{subset} If $M\subseteq \mathcal{P}$ is finite and $\phi$ is
a formula then
\[ M \mbox{ represents }\phi\mbox{ iff }\phi\equiv\bigvee_{p\in M}\chi(p).\]
\end{enumerate}
\end{lemma}

\begin{proof}
\begin{enumerate}
\item It follows directly from the definitions of these two concepts and
Lemma \ref{lem:char-form}.

\item For any $q\in \mathcal{P}$ we proceed as follows:
\[
\exists p\in M.p\cclt q\Leftrightarrow
\exists p\in M.q\models\chi(p)\Leftrightarrow
q\models\bigvee_{p\in M}\chi(p).
\]
Now the statement of the lemma follows easily from this fact and Definition
\ref{def:rep}.\qed
\end{enumerate}
\end{proof}

We want to characterize the set of formulae that can be represented by
a finite set of processes, and in particular by a single process. For
this purpose we introduce some notions of normal form for logical
formulae.

\begin{definition}
\begin{enumerate}
\item A formula $\phi$ is in {\sl normal form} if it has the form
\[
\phi=\bigvee_{i\in I}(\bigwedge_{j\in J_i}\may{a_j^i}\phi_j^i\land
\bigwedge_{k\in K_i}\must{b_k^i}\psi_k^i).
\]
where all $\phi_j^i$ and $\psi_k^i$ are also in normal form. In
particular, $\bot$ is obtained when $I=\emptyset$ and $\top$ when
$I=\{1\}$ and $J_1=K_1=\emptyset$.

\item A formula $\psi$ is in {\sl strong normal form} if it has the form
\[
\psi=\bigvee_{i\in I}\phi_i\;,
\]
where each $\phi_i$ is in unary strong normal form.  A formula $\phi$
is in {\sl unary strong normal form} if it is $\top$ or it has the form
\[
\phi=\bigwedge_{j\in J}\may{a_j}\phi_j\land
\bigwedge_{b\in A^l}\must{b}\psi_b,
\]
where every  $\phi_j$ is in unary strong normal form and every $\psi_b$ is
in strong normal form.

\end{enumerate}
\end{definition}

We note that any unary strong normal form different from $\top$ can
equivalently be written as
\[
\phi=\bigwedge_{j\in J}\may{a_j}\phi_j\land
\bigwedge_{b\in A^l}\must{b}\bigvee_{k\in {K_b}}\psi^k_b,
\]
where every $\phi_j$ and every $\psi^k_b$ are in unary strong normal
form, thus avoiding the introduction of strong normal forms.

\begin{remark}
\emph{It is not hard to see that each unary strong normal form is
consistent. See also Theorem~\ref{thm:unarynormalform} to follow.}
\end{remark}

Clearly the characteristic formulae of processes are in unary strong
normal form. Therefore, by Lemma~\ref{lem:char1}, it is a necessary
condition for a formula to be representable by a single process that
it has an equivalent unary strong normal form. We will show that this
is also a sufficient condition for this to hold for any consistent
formula.
\begin{theorem}\label{thm:unarynormalform}
A unary strong normal form
\[
\phi=\bigwedge_{j\in J}\may{a_j}\phi_j\land \bigwedge_{b\in
  A^l}\must{b}\bigvee_{k\in K_b}\psi^k_b
\]
is represented by the process defined recursively by
\begin{eqnarray*}
\theta(\phi) & = & \sum_{j\in J}a_j.\theta(\phi_j)+\sum_{b\in
    A^l}\sum_{k\in K_b}b.\theta(\psi^k_b),\quad\textrm{if
    $\phi\neq\top$} \\
\theta(\top) & = & \omega.
\end{eqnarray*}

In particular $\phi$ is the characteristic formula for $\theta(\phi)$
(up to logical equivalence). Note that even if in the formal
expression above there is a summand for each $b\in A^l$, only those
$b$'s such that $K_b\neq\emptyset$ will finally appear as summands of
$\theta(\phi)$.
\end{theorem}

\begin{proof}
  First we prove that $\theta(\phi)\models\phi$ by induction on the
  modal depth of $\phi$. If $\phi=\top$ we have that obviously
  $\theta(\phi)=\omega\models\phi=\top$. For the inductive step first
  we note that $\theta(\phi)\der{a_j}\theta(\phi_j)$ for all $j\in
  J$. By induction, $\theta(\phi_i)\models\phi_i$. Next assume that
  $\theta(\phi)\der{b}p$ for some $b\in A^l$ and some $p$. We have
  that $p=\theta(\psi^k_b)$ for some $k\in K_b$. By induction
  $\theta(\psi^k_b)\models\psi^k_b$ and therefore
  $\theta(\psi^k_b)\models\bigvee_{k\in K_b}\psi^k_b$.

  Next we prove that if $q\models\phi$ then $\theta(\phi)\cclt
  q$. Towards proving this claim, assume that $q\models\phi$. Again we proceed
  by induction on the modal depth of $\phi$.

First assume that $\theta(\phi)\der{a}p'$ for some $a\in A^r$ and
  process term $p'$. Then $a=a_j$ for some $j\in J$ and
  $p'=\theta(\phi_j)$. As $q\models\phi$, we have that $q\der{a_j}q'$ for some
  $q'$ with $q'\models\phi_j$. By induction, $\theta(\phi_j)\cclt
  q'$, as required.

  Now assume that $q\der{b}q'$ for some $b\in A^l$.  As
  $q\models\phi$ we have that $q'\models\psi^k_b$ for some $k\in K$.  Now
  $\theta(\phi)\der{b}\theta(\psi^k_b)$ and, by the induction hypothesis,
  we have $\theta(\psi^k_b)\cclt q'$, as required.

  This proves that $\phi$ is the characteristic formula for
  $\theta(\phi)$ and therefore, by Lemma \ref{lem:char1}, that
  $\theta(\phi)$ represents $\phi$.\qed
\end{proof}
Next, we will show that any formula has an equivalent strong normal
form and therefore can always be represented by a (possibly empty)
finite set of processes. To derive this result we will use several
standard equivalences between formulae.

\begin{lemma}\label{equalities}
The following statements hold.
\begin{enumerate}
\item $\land$ and $\lor$ are associative, commutative and idempotent.

\item $\land$ distributes over $\lor$, and $\lor$ distributes over $\land$.

\item $\phi\lor\top\equiv \top$, $\phi\lor\bot\equiv \phi$, $\phi\land\top\equiv \phi$, and
$\phi\land\bot\equiv \bot$.

\item $\must{b}\top\equiv \top$.

\item $\must{b}\phi\land\must{b}\psi\equiv \must{b}(\phi\land\psi)$ for
$b\in A^l$.

\item $\may{a}\phi\lor\may{a}\psi\equiv \may{a}(\phi\lor\psi)$ for
$a\in A^r$.

\end{enumerate}
\end{lemma}

\begin{proof}
The first three collections of equalities are straightforward and well
known, so we omit their proofs.
\begin{itemize}
\item $\must{b}\top\equiv \top$. We have $p\models\must{b}\top$ iff
$p'\models\top$ for all $p\der{b}p'$. Therefore, the condition is
satisfied whenever $p\der{b}p'$, and it is vacuously true when
$p\stackrel{b}{\not\der{}}$.

\item $\must{b}\phi\land\must{b}\psi\equiv
\must{b}(\phi\land\psi)$. We have
$p\models(\must{b}\phi\land\must{b}\psi)$ iff $p'\models\phi$ for all
$p\der{b}p'$ and $p'\models\psi$ for all $p\der{b}p'$, iff
$p'\models(\phi\land\psi)$ for all $p\der{b}p'$, iff
$p\models\must{b}(\phi\land\psi)$.

\item $\may{a}\phi\lor\may{a}\psi\equiv \may{a}(\phi\lor\psi)$. We have
  $p\models\may{a}\phi\lor\may{a}\psi$ iff there exists $p\der{a}p'$
  such that $p'\models\phi$ or there exists $p\der{a}p''$ such that
  $p''\models\psi$, that is, iff there exists some $p\der{a}p'_0$ such
  that $p'_0\models\phi$ or $p'_0\models\psi$. This holds iff
  $p\models\may{a}(\phi\lor\psi)$. \qed
\end{itemize}
\end{proof}
\begin{lemma}\label{lem:normalform}
Every formula $\phi$ has an equivalent strong normal form with no
larger modal depth.
\end{lemma}

\begin{proof}
  First we prove by induction on the modal depth, using $1$-$3$ of
  Lemma~\ref{equalities}, that $\phi$ has an equivalent normal form
  with the same modal depth. To prove the main statement we can
  therefore assume that $\phi$ is in normal form. We proceed by
  induction on the modal depth $md(\phi)$.  The base case $md(\phi)=0$
  ($\phi \equiv\bot$ and $\phi \equiv \top$) follows immediately.

 Next let us assume that
\[
\phi=\bigvee_{i\in I}(\bigwedge_{j\in J_i}\may{a_j^i}\phi_j^i\land
\bigwedge_{k\in K_i}\must{b_k^i}\psi_k^i).
\]
By Lemma~\ref{equalities}, using $4$ and $5$ and the standard laws described in $1$-$3$, $\phi$
can be rewritten into an equivalent formula of the form
\[
\phi=\bigvee_{i\in I}(\bigwedge_{j\in J_i}\may{a_j^i}\phi_j^i\land
\bigwedge_{b\in A^l}
\must{b}\psi_b^i)
\]
where $md(\psi_b^i)\leq \sup\{md(\psi_k^i)\mid k\in K_i\}$ (we note that
some of the $\must{b}\psi^i_b$s may have the form $\must{b}\top$,
which is equivalent to $\top$). Therefore, by the induction
hypothesis, we may assume that $\phi_j^i$ and $\psi_b^i$ are in strong
normal form.  Next we use Lemma~\ref{equalities}.$6$ to remove all the
occurrences of $\lor$ that are guarded by $\may{a}$, for some $a\in
A^r$ in each $\bigwedge_{j\in J_i}\may{a_j^i}\phi_j^i$. The result for
each $i$ is of the form $\bigwedge_{j\in J_i}(\bigvee_{l\in
  L_j}\may{a_j^i}\phi_j^{l,i})$, where each $\phi_j^{l,i}$ is in a
unary strong normal form. By repeated use of distributivity, the whole
formula can be rewritten as
\[
\phi=\bigvee_{r\in R}(\bigwedge_{s\in S_r}\may{a^r_s}\alpha^r_s\land
\bigwedge_{b\in A^l}\must{b}\bigvee_{t\in T_b^r}\beta_{b,t}^r)
\]
where each $\alpha_r^s$ and $\beta^{b,t}_r$ is a unary strong normal
form.  Finally we note that the operations described above do not
increase the modal depth.  \qed
\end{proof}

Now we will relate our result to the one in Boudol and Larsen's
paper~\cite{BoudolL1992}.

\begin{definition}
A formula $\phi$ is {\em prime} if the following holds:
\[
\forall \phi_1,\phi_2\in\Lo.~\phi\leq\phi_1\lor\phi_2
\mbox{ implies }\phi\leq\phi_1\mbox{ or }\phi\leq\phi_2.
\]
\end{definition}

\begin{theorem}
A formula  $\phi$ can always be represented by a finite set of
processes. It can be represented by a single process if and only if it
is consistent and prime.
\end{theorem}

\begin{proof}
  By Lemma \ref{lem:normalform}, $\phi\equiv\phi_1\lor\ldots\lor\phi_n$
  where each $\phi_i$, $1\leq i\leq n$, is in unary strong normal form.
  By Theorem \ref{thm:unarynormalform}, $\phi_i\equiv\chi(p_i)$ for
  some $p_i$ for each $1\leq i\leq n$, and therefore
  $\phi\equiv\chi(p_1)\lor\ldots\lor\chi(p_n)$. The first statement
  now follows from Lemma~\ref{lem:char1}.\ref{subset}.

Towards proving the second statement, first assume that
$\phi\equiv\chi(p_1)\lor\ldots\lor\chi(p_n)$ is prime.  This implies
that $\phi\leq\chi(p_i)\leq\phi$, for some $i\in \{1,\ldots, n\}$,
which in turn implies that $\phi \equiv \chi(p_i)$.

Next assume that $\phi$ is represented by some process $p$ or
equivalently that $\phi\equiv\chi(p)$. Now assume that
$\chi(p)\leq\phi_1\lor\phi_2$. As $p\models\chi(p)$, this implies that
$p\models\phi_1\lor\phi_2$ or equivalently that either
$p\models\phi_1$ or $p\models\phi_2$. Without loss of generality, we
can assume that $p\models\phi_1$. Now assume that $r\models\chi(p)$.
Then $p\cclt r$ and by Theorem~\ref{thm:log} this implies that
$r\models\phi_1$.  Since $r$ was arbitrary, this proves that
$\phi\equiv\chi(p)\leq\phi_1$. Hence $\phi$ is prime, which was to be shown.
\hfill$\Box$
\end{proof}

\section{Considering bivariant actions}\label{sec:bivar}

Originally~\cite{FabregasFP09,FabregasEtAl10-sos,FabregasEtAl10-logics},
the theory of covariant-contravariant semantics also considered
bivariant actions in $A^\mathit{bi}$, so that we had a partition of
$A$ into $\{A^r,A^l,A^\mathit{bi}\}$ (called the signature of the
LTS), and the definition of covariant-contravariant simulations
imposed the following two conditions:

\begin{itemize}
\item For all $a\in A^r\cup A^\mathit{bi}$ and all $p\der{a} p'$, there
 exists some $q\der{a}q'$ with $p' \mathrel{R} q'$.

\item For all $a\in A^l\cup A^\mathit{bi}$ and all $q\der{a}q'$, there
 exists some $p\der{a}p'$ with $p' \mathrel{R} q'$.
\end{itemize}

When we have in our signature bivariant actions we cannot get directly the graphical representation results that we have presented in Section~\ref{sec:graph}. This is so because bivariant actions cannot be under approximated, as a consequence of the well known result that bisimilarity is an equivalence relation and not a plain preorder. In order to maintain our results we mandatorily need that notion of approximation. We obtain it by decomposing each bivariant action $a$ into a pair of actions, one covariant, $a^r$, and another contravariant, $a^l$. Technically, we define an embedding of the set of processes over an arbitrary signature $A=\{A^r,A^l,A^\mathit{bi}\}$ into that corresponding to a new signature $\bar{A}=\{\bar{A}^r,\bar{A}^l,\emptyset\}$. The latter does not include any bivariant action, and then we can apply to it our graphical representation results, that then can be transfered to the original signature by means of the defined embedding.

In~\cite{AcetoEtAl11} we presented transformations from LTSs to Modal Transition Systems (MTSs), and vice versa, named $\calM$ and $\calC$, respectively. We proved that both preserve and reflect the covariant-contravariant logic and simulation preorder. Applying these two transformations in a row we did not obtain the identity function, but instead a transformation $\mathcal{T}_0=\calC\com\calM$ that transforms an LTS with bivariant actions into another LTS without them. Since composition preserves the good properties of $\calC$ and $\calM$, $\mathcal{T}_0$ also has these properties. 

Next we give a direct definition of $\mathcal{T}_0$.

\begin{definition}\label{def:cm}
Let $T$ be an LTS with the signature $A=\{A^r,A^l,A^\mathit{bi}\}$. The LTS $\mathcal{T}_0(T)$ with signature $\hat{A}=\{\hat{A}^r,\hat{A}^l,\emptyset\}$, where $\hat{A}^r=\{d^r\mid d\in A^r\cup A^\mathit{bi}\}$ and $\hat{A}^l=\{d^l\mid d\in A^r\cup A^l\cup A^\mathit{bi}\}$, is constructed as follows:
\begin{itemize}
\item The set of states of $\mathcal{T}_0(T)$ is the same as the one of $T$ plus a new
 state $u$.

\item For each transition $p\der{d} p'$ in $T$ , add a transition 
 $p \der{d^l} p'$ in $\mathcal{T}_0(T)$.

\item For each transition $p \der{d} p'$ in $T$ with $d\in A^r \cup A^{bi}$,
 add a transition $p\der{d^r}p'$ in $\mathcal{T}_0(T)$. 

\item For each $a\in A^r$ and state $p$, add the transition 
 $p\der{a^l} u$ to $\mathcal{T}_0(T)$, as well as transitions 
 $u \der{d^l} u$, for each action $d\in A$.   
\end{itemize}
\end{definition}

Note that each $c\in A^\mathit{bi}$ is ``encoded'' by means of a pair of new actions $(c^r, c^l)$. Moreover, as a consequence of the general definition of $\calM$, for each $a\in A^r$, together with $a^r$, which is its ``natural'' encoding an additional $a^l\in A^l$, coupled with it, is introduced. Finally, the behaviour of the ``extra'' state $u$ is defined by $\omega$.

\begin{figure}[htp]
\centering
\fbox{
\begin{minipage}{0.8\linewidth}
\phantom{aaaaaaaaaaaaaaaa\\aaaa}
\bigskip
\begin{center}
$\begin{array}{c@{\hskip 1cm}c@{\hskip 1cm}c} {\begin{array}{c@{\hskip 1.5cm}c@{\hskip 1.5cm}c} \rnode{a}{X}  & \rnode{b}{Y} & \rnode{c}{Z}
\end{array}
\Large \psset{nodesep=3pt}\everypsbox{\scriptstyle} \ncarc{->}{a}{b}\Aput{a} \ncarc{->}{b}{a}\Aput{c} \ncline{->}{c}{b}\Bput{b} }

& {\stackrel{{\calC}\mathop{\circ}{\calM}}{\longmapsto}} &

{\begin{array}{c@{\hskip 1.8cm}c@{\hskip 1.8cm}c}
\rnode{a}{X}  & \rnode{b}{Y} & \rnode{c}{Z}\\[1.8cm]
& \rnode{d}{u} &
\end{array}
\Large \psset{nodesep=3pt} \everypsbox{\scriptstyle} \ncarc{->}{a}{b}\Aput{a^r,a^l} \ncarc{->}{b}{a}\Aput{c^r,c^l} \ncline{->}{c}{b}\Bput{b^l}
\ncline{->}{a}{d}\Bput{a^l} \ncline{->}{b}{d}\Aput{a^l} \ncline{->}{c}{d}\Aput{a^l}\nccircle[angleA=180]{->}{d}{.35cm}\Bput{a^l,c^l,b^l}}
\end{array}
$
\end{center}
\bigskip \caption{The original transformation of a LTS with bivariant actions into another without them, assuming $A^r=\{a\}$,
$A^l=\{b\}$ and $A^\mathit{bi}=\{c\}$.}\label{fig:1}
\end{minipage}}
\end{figure}

Based on this transformation, we have designed a direct encoding of LTSs over a signature $A=\{A^r,A^l,A^\mathit{bi}\}$
by means of LTSs over an adequate signature $\bar{A}=\{\bar{A}^r,\bar{A}^l,\emptyset\}$. As above, for each $c\in A^\mathit{bi}$ in the original signature, we introduce a pair of (new) actions, as the following definition makes precise.

\begin{definition}\label{Def:Bi2nBi}
Let $T$ be an LTS with signature $A=\{A^r, A^l, A^\mathit{bi}\}$. The LTS
$\tbi{T}$, with signature $\bar{A}=\{\bar{A}^r, \bar{A}^l,\emptyset\}$, where
$\bar{A}^r = A^r\cup\{c^r \mid c\in A^\mathit{bi}\} $ and $\bar{A}^l
=A^l\cup\{c^l \mid c\in A^\mathit{bi}\}$, is constructed as follows:
\begin{itemize}
\item The set of states of $\tbi{T}$ is the same as that of $T$.

\item All the transitions from $T$ with label in $A^r \cup A^l$ are in $\tbi{T}$.

\item For each transition $p\der{c} p'$ in $T$ with $c\in A^\mathit{bi}$, we add
 $p\der{c^r} p'$ and $p\der{c^l} p'$ to $\tbi{T}$.
\end{itemize}
\end{definition}

The transformation above produces an LTS without bivariant actions more closely related to the original covariant-contravariant LTS than that produced by $\mathcal{T}_0$ (compare Figure~\ref{fig:2} with Figure~\ref{fig:1}). Note that the class of LTSs with signature $\bar{A}$ that satisfy that $p\der{c^r}p'$ if and only if $p\der{c^l}p'$, for all $p,p'\in\Proc$, and all $c\in A^\mathit{bi}$; is exactly the class of processes that are the representation of some LTS with signature $A$.

To translate modal formulae we have just to adopt the right modality for each action, as the
following definition makes precise.

\begin{definition}
Let us extend $\mathcal{T}$ to translate modal formulae over the modal logic for LTS over $A$ into modal formulae over the modal logic for LTS over $\bar{A}$, as follows:
\begin{itemize}
\item $\tbi{\bot} = \bot$.
\item $\tbi{\top} = \top$.
\item $\tbi{\varphi\land\psi} = \tbi{\varphi}\land \tbi{\psi}$.
\item $\tbi{\varphi\lor\psi} = \tbi{\varphi}\lor \tbi{\psi}$.
\item $\tbi{\langle a\rangle\varphi} = \langle a\rangle \tbi{\varphi}$, if $a\in A^r$.
\item $\tbi{\langle c\rangle\varphi} = \langle c^r\rangle \tbi{\varphi}$, if $c\in A^\mathit{bi}$.
\item $\tbi{[b]\varphi} = [b]\tbi{\varphi}$, if $b\in A^l$.
\item $\tbi{[c]\varphi} = [c^l]\tbi{\varphi}$, if $c\in A^\mathit{bi}$.
\end{itemize}
\end{definition}

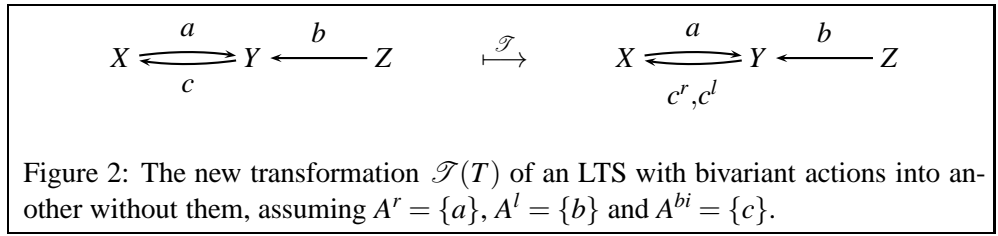
\begin{figure}[htb]
\centering
\fbox{
\begin{minipage}{0.8\linewidth}
\phantom{aaaaaaaaaaaaaaaaaaaa}
\begin{center}
$\begin{array}{c@{\hskip 1cm}c@{\hskip 1cm}c} {\begin{array}{c@{\hskip 1.5cm}c@{\hskip 1.5cm}c} \rnode{a}{X}  & \rnode{b}{Y} & \rnode{c}{Z}
\end{array}
\Large \psset{nodesep=3pt}\everypsbox{\scriptstyle} \ncarc{->}{a}{b}\Aput{a} \ncarc{->}{b}{a}\Aput{c} \ncline{->}{c}{b}\Bput{b} }

& {\stackrel{\mathcal{T}}{\longmapsto}} &

{\begin{array}{c@{\hskip 1.5cm}c@{\hskip 1.5cm}c} \rnode{a}{X}  & \rnode{b}{Y} & \rnode{c}{Z}
\end{array}
\Large \psset{nodesep=3pt} \everypsbox{\scriptstyle} \ncarc{->}{a}{b}\Aput{a} \ncarc{->}{b}{a}\Aput{c^r,c^l} \ncline{->}{c}{b}\Bput{b}}
\end{array}
$
\end{center}
\bigskip \caption{The new transformation $\tbi{T}$ of an LTS with bivariant actions into another without them, assuming $A^r=\{a\}$,
$A^l=\{b\}$ and $A^\mathit{bi}=\{c\}$.}\label{fig:2}
\end{minipage}}
\end{figure}

In order to show that $\mathcal{T}$ preserves and reflects the cc-simulation preorder, we compare $\tbi{T}$ with $\mathcal{T}_0(T)$ and we prove a more general result.

\begin{definition}
Given a signature $\{A^r,A^l,\emptyset\}$ and $c^l\in A^l$ we define the transformation $\mathcal{T}^+_{c^l}$ as that which given an LTS $T$ with that signature adds a new state $u$ whose behaviour is that defined by $\omega$, and a new transition labelled by $c^l$ from each state of $T$ to $u$.
\end{definition}

\begin{proposition}\label{prop:t+}
$\mathcal{T}^+_{c^l}$ preserves and reflects the cc-simulation preorder when applied to a system that does not contain any $c^l$ transition.
\end{proposition}

\begin{proof}
We will see that $R$ is a cc-simulation in $T$ if and only if $R\cup\{(u,u)\}$ is a cc-simulation in $\tabi{T}$. The result is immediate by simply observing that for $a$-transitions, with $a\neq c^l$, the leaving of any state $p$ with $p\neq u$ are exactly the same in $T$ and $\tabi{T}$, while for any such state we always have $p\der{c^l}u$ in $\tabi{T}$.\qed
\end{proof}

\begin{corollary}\label{cor:t+}
Let $T$ be an LTS with signature $\{A^r,A^l,A^\mathit{bi}\}$. Then, for any two states $p$ and $q$ of $T$, we have $p\ccsim q$ in $\tbi{T}$ if and only if ${{p}\ccsim {q}}$ in ${\mathcal{T}_0(T)}$.
\end{corollary}

\begin{proof}
Note that $\tbi{T}$ is a $\{\bar{A}^r,\bar{A}^l,\emptyset\}$-LTS, while $\mathcal{T}_0(T)$ is an $\{\hat{A}^r,\hat{A}^l,\emptyset\}$-LTS,
where $\hat{A}^r=\{a^r\mid a\in A^r\cup A^\mathit{bi}\}$ and $\hat{A}^l=\bar{A}^l\cup\{a^l\mid a\in A^r\}$. This means that we can also see 
$\tbi{T}$ as an $\{\hat{A}^r,\hat{A}^l,\emptyset\}$-LTS if we rename each $a\in A^r$ into the corresponding $a^r\in\hat{A}^r$. Then, we can apply 
$\mathcal{T}^+_{a^l}$ for each $a\in A^r$ in a row, thus getting a transformed system $\mathcal{T}^+(T)$. All along these applications we are 
under the hypothesis of Proposition~\ref{prop:t+}. Moreover, the only differences between $\mathcal{T}^+(T)$ and $\mathcal{T}_0(T)$ are the 
collection of $a^l$-transitions paired with the $a^r$-transitions in $T$, with $a\in A^r$. But since for any state $p$ of $\mathcal{T}^+(T)$ we 
have $p\der{a^l}u$, for all $a^l\in\{a^l\mid a^r\in A^r\}$, we immediately conclude that the identity is a cc-simulation in both directions (up-to the indicating renaming) between the states of $\mathcal{T}^+(T)$ and those in $\mathcal{T}_0(T)$, from which we finally obtain that 
$p\ccsim q$ in $\tbi{T}$ iff $p\ccsim q$ in $\mathcal{T}_0(T)$.\qed
\end{proof}

\begin{corollary}
Our transformation $\mathcal{T}$ preserves and reflects the cc-simulation preorder, that is, for each LTS $T$ and for all states $p$ and $q$ in $T$, it holds that $p\ccsim q$ in $T$ if, and only, if $p\ccsim q$ in $\tbi{T}$.
\end{corollary}

\begin{proof}
We just need to combine Proposition~\ref{prop:t+} and Corollary~\ref{cor:t+}.\qed
\end{proof}

\begin{proposition}
$\mathcal{T}$ preserves and reflects the cc-logic, that is, for each LTS $T$, any state $p$ and all covariant-contravariant formula $\varphi$ in $T$, it holds that $p\models\varphi$ in $T$ if, and only if, $p\models\tbi{\varphi}$ in $\tbi{T}$.
\end{proposition}

\begin{proof}
We proved in~\cite{AcetoEtAl11} the corresponding result for $\mathcal{T}_0$ and the transformation $\mathcal{T}_0$ which is defined on logic formulae exactly as $\mathcal{T}$, but renaming again each $a\in A^r$ into $a^r$. From the definitions of $\mathcal{T}$ and $\mathcal{T}_0$ we immediately conclude that $a^l$-transitions with $a\in A^r$ do not play any role in the satisfaction of any formula $\tbi{\varphi}$, and then the result follows from that proved in~\cite{AcetoEtAl11}.\qed
\end{proof}

After the representation of a bivariant action $c\in A^\mathit{bi}$ as a pair $(c^r,c^l)$ with $c^r\in\bar{A}^r$ and $c^l\in\bar{A}^l$, we have that $c^l$ under-approximates $c$, whereas $c^r$ over-approximates $c$. This means in particular that we have $c^l0\ccsim c^l0+c^r0\ccsim c^r0$ and, more generally, $c^lp\ccsim c^lp+c^rq\ccsim c^rq$, for all processes $p$ and $q$. Therefore, once we have separated the covariant and contravariant characters of bivariant actions we achieve a greater flexibility which allows us to consider ``non-balanced'' processes where these two characters do not go always together, thus producing over and under-approximations when needed.

\begin{discussion}
It is interesting to compare our new transformation $\mathcal{T}$ with the original transformation $\mathcal{T}_0$ from~\cite{AcetoEtAl11}. The first aims to obtain a representation over the signature $\{\bar{A}^r,\bar{A}^l,\emptyset\}$ that is as simple as possible, and this is why we do not introduce $a^l$ when $a\in A^r$. Instead, we can see the result of the transformation $\mathcal{T}_0$ as a process in the ``uniform'' signature $\tilde{A}=\{\tilde{A}^r,\tilde{A}^l,\emptyset\}$, with $\tilde{A}^r=\{a^r\mid a\in A^r\cup A^l\cup A^\mathit{bi}\}$ and $\tilde{A}^l=\{a^l\mid a\in A^r\cup A^l\cup A^\mathit{bi}\}$. It is true that the actions $b^r$ with $b\in A^l$ do not appear in $\mathcal{T}_0(T)$, but even so we can consider any $\mathcal{T}_0(T)$ as a process for $\tilde{A}$. Obviously, this is also the case for $\tbi{T}$, where the actions $a^l$ with $a\in A^r$ do not appear either. Both $\mathcal{T}_0(T)$ and $\tbi{T}$ were ``good'' representations of $T$, as stated above, however it is clear that we do not have $\mathcal{T}_0(T)\equiv_{cc}\tbi{T}$. Instead, $\mathcal{T}_0(T)\ccsim\tbi{T}$, and in fact $\mathcal{T}_0(T)$ is the least process with respect to $\ccsim$, for the uniform signature $\tilde{A}$ that has the good properties stated in the paper. Note that, instead, $b^r$-transitions for $b\in A^l$ do not need to be introduced at all, since any addition of a covariant transitions produces a $\ccsim$-greater process.

Therefore, the original transformation $\mathcal{T}_0$ would be indeed the adequate one if we wanted to obtain an embedding of the class of processes for any signature into that corresponding to the uniform signature $\tilde{A}$ defined above, where all the actions can be interpreted as the covariant and contravariant parts of the actions in a set $A$.
\end{discussion}

To conclude the section we explore the set of systems for any signature $\bar{A}=\{\bar{A}^r,\bar{A}^l,\emptyset\}$. Some of them, but not all, are equivalent to the representation
of a system for the original alphabet $A$. Whenever that is not the case we would need to remove (or add) some transitions labelled by the created
actions in $\{c^r,c^l\mid c\in A^\mathit{bi}\}$ in order to obtain a system that is equivalent to the representation of some process. In the following proposition we give an algorithm for obtaining a system for the original signature $A$ to which a given system for the signature $\bar{A}$ is equivalent, whenever such a system exists. To make possible a proof by (structural) induction, we will only present the result for process terms in $\mathcal{P}$.

\begin{proposition}
Let $A=\{A^r,A^l,A^\mathit{bi}\}$ be a signature and $\bar{A}=\{\bar{A}^r,\bar{A}^l,\emptyset\}$ be the associated signature without bivariant actions. Let
$p,q\in\mathcal{P}$ be process terms for $\bar{A}$ such that $q$ is the representation of some process for the signature $A$. Let us assume that $p\equiv_{cc}q$. Then it is possible to transform $p$ into the representation $p_{bi}$ of some process term for $A$, simply by adding or
removing some transitions labelled by actions in $\{c^r,c^l\mid c\in A^\mathit{bi}\}$.
\end{proposition}

\begin{proof}
The proof is done by structural induction.
\begin{itemize}
\item If $p=0$ or $p=\omega$ we can take $p_{bi}=p$.

\item In the general case, we exploit the fact that whenever $a\in\bar{A}^r$, if $q'\ccsim p'$ then $ap'+aq'\equiv_{cc} ap'$ (and dually, when $b\in\bar{A}^l$, $bp'+bq'\equiv_{cc} bq'$). This means that from any term for $\bar{A}$ we can remove all the summands $aq''$ (resp. $bp''$) such that $ap''$ is not a maximal $a$-summand of $p'$ with respect to $\ccsim$ (resp. $bp''$ is not a minimal $a$-summand), obtaining a $\equiv_{cc}$-equivalent process. So, we start by removing all the non-maximal $a$-summands with $a\in\bar{A}^r$, and all the non-minimal $b$-summands with $b\in\bar{A}^l$ of any subterm of $p$. By abuse of notation, we will still denote the obtained process by $p$, and we still have $p\equiv_{cc}q$.

Now, for any $a$-summand of $p$ with $a\in\bar{A}^r$, $p=p'+ap''$, there is some $q\der{a}q''$ with $p''\ccsim q''$. But also, since $p\equiv_{cc}q$, starting with $q\der{a}q''$ there must exist some $p\der{a}p'''$ with $q''\ccsim p'''$, but then $p''\ccsim p'''$, and since $p''$ was maximal we can assume that $p'''=p''$, and then we also have $p''\equiv_{cc}q''$. The same is true for all the $b$-summands with $b\in\bar{A}^l$, and this means that we can apply the induction hypothesis to all the derivatives of $p$.

Moreover, for each $ap'$ summand with $a=c^r$ we can add to $p$ the summand $c^lp'$ and we obtain $p\equiv_{cc} p+c^lp'$. Indeed, we have trivially $p+c^lp'\ccsim p$, and to prove that $p\ccsim p+c^lp'$ we check $q\ccsim p+c^lp'$. We only need to see that for any transition $p+c^lp'\der{c^l} p'$ there is some $q\der{c^l}q'$ with $q'\ccsim p'$. We use again the maximality of the summand $c^rp'$ and we obtain, as above, that there is some $c^rq'$ summand of $q$ with $q'\ccsim p'$. But since $q$ was the representation of some process for $A$, it has also a summand $c^lq'$ as required above.

The obtained process has already its $c^r$ and $c^l$ transitions, with $c\in A^\mathit{bi}$, paired at its first level, and then we simply need to apply the induction hypothesis to conclude the proof.\qed
\end{itemize}
\end{proof}

\begin{remark}
\emph{Although the proposition above assumes that the considered process was equivalent to the representation of some process for $A$, it is easy to use it as a decision algorithm to check that property: we apply the algorithm to the given process $p$ and check if the obtained process $p'$ is $\equiv_{cc}$-equivalent to it, if that is not the case then $p$ is not equivalent to the representation of any process for the signature $A$.}

\end{remark}

\section{Conclusions and future work}\label{sec:future}

In~\cite{AcetoEtAl11} we studied the relationships between the notion
of refinement over modal transition systems, and the notions of
covariant-contravariant simulation and partial bisimulation over
labelled transition systems. Here we have continued that work by
looking for the ``graphical'' representation of the
covariant-contravariant modal formulae by means of terms, as it was
done in~\cite{Baetenetal} for the case of modal transition
systems. For technical reasons, we had first to restrict ourselves to
the case in which we have no bivariant actions. Afterwards, we argued
that the general case can, in some sense, be ``reduced'' to the one we
dealt with in Section~\ref{sec:graph} by defining a
semantic-preserving transformation between covariant-contravariant
systems with bivariant actions, and covariant-contravariant systems
without them.

The idea was to separate each bivariant action into its covariant and
its contravariant parts. As a matter of fact, we believe that this
idea might be useful not only for obtaining theoretical results, as we
have done here, but also for applications. Most of the studies on
process algebras and their semantics assume the bivariant behaviour of
all the actions. It is true that in some studies (see for
example~\cite{Lynch88}) we have a classification of actions, as we
have also done in~\cite{AcetoEtAl11} and in this paper. But now we are
proposing to exploit the relationships between the different classes
of actions.

As future work, it would be interesting to obtain a direct
characterization of the formulae that are graphically representable in
a setting with bivariant actions. Such a direct characterization will
also pave the way towards a more general theory of ``graphical
characterizations'' of formulae in modal logics of processes, of which
the result by Boudol and Larsen and ours are special cases.

Of course, one of the directions in which we plan to continue our studies is that related with the logical characterization of the semantics,
and in particular the connections between logical formulae and terms established by characteristic formulae and graphical representations. The
combination of these two frameworks is also an interesting challenge. In particular, we plan some extensions of the recent work by L{\"u}ttgen
and Vogler~\cite{Vogler09,Vogler10} to the case of covariant-contravariant systems.


\begin{thebibliography}{10}
\providecommand{\bibitemdeclare}[2]{}
\providecommand{\urlprefix}{Available at }
\providecommand{\url}[1]{\texttt{#1}}
\providecommand{\href}[2]{\texttt{#2}}
\providecommand{\urlalt}[2]{\href{#1}{#2}}
\providecommand{\doi}[1]{doi:\urlalt{http://dx.doi.org/#1}{#1}}
\providecommand{\bibinfo}[2]{#2}

\bibitemdeclare{inproceedings}{AcetoEtAl11}
\bibitem{AcetoEtAl11}
\bibinfo{author}{Luca Aceto}, \bibinfo{author}{Ignacio F\'abregas},
  \bibinfo{author}{David de~Frutos~Escrig}, \bibinfo{author}{Anna
  Ing\'olfsd\'ottir} \& \bibinfo{author}{Miguel Palomino}
  (\bibinfo{year}{2011}): \emph{\bibinfo{title}{Relating modal refinements,
  covariant-contravariant simulations and partial bisimulations}}.
\newblock In {\sl \bibinfo{booktitle}{Fundamentals of Software
  Engineering, FSEN 2011}}, \bibinfo{series}{LNCS}, \bibinfo{publisher}{Springer}.\bibinfo{notes}{To appear}.

\bibitemdeclare{book}{AcetoEtAl07b}
\bibitem{AcetoEtAl07b}
\bibinfo{author}{Luca Aceto}, \bibinfo{author}{Anna Ing\'olfsd\'ottir},
  \bibinfo{author}{Kim~Guldstrand Larsen} \& \bibinfo{author}{Ji\^{r}\'{i}
  Srba} (\bibinfo{year}{2007}): \emph{\bibinfo{title}{Reactive Systems:
  Modelling, Specification and Verification}}.
\newblock \bibinfo{publisher}{Cambridge University Press}.

\bibitemdeclare{techreport}{Baetenetal}
\bibitem{Baetenetal}
\bibinfo{author}{J.~Baeten}, \bibinfo{author}{D.~van Beek},
  \bibinfo{author}{B.~Luttik}, \bibinfo{author}{J.~Markovski} \&
  \bibinfo{author}{J.~Rooda} (\bibinfo{year}{2010}):
  \emph{\bibinfo{title}{Partial Bisimulation}}.
\newblock \bibinfo{type}{SE Report} \bibinfo{number}{2010-04},
  \bibinfo{institution}{Department of Mechanical
  Engineering, Eindhoven University of Technology}, \url{http://se.wtb.tue.nl/sereports}.

\bibitemdeclare{article}{BoudolL1992}
\bibitem{BoudolL1992}
\bibinfo{author}{G{\'e}rard Boudol} \& \bibinfo{author}{Kim~Gulstrand Larsen}
  (\bibinfo{year}{1992}): \emph{\bibinfo{title}{Graphical versus logical
  specifications}}.
\newblock {\sl \bibinfo{journal}{Theoretical Computer Science}}
  \bibinfo{volume}{106}(\bibinfo{number}{1}), pp. \bibinfo{pages}{3--20},
  \doi{10.1016/0304-3975(92)90276-L}.

\bibitemdeclare{inproceedings}{FabregasFP09}
\bibitem{FabregasFP09}
\bibinfo{author}{Ignacio F{\'a}bregas}, \bibinfo{author}{David
  de~Frutos-Escrig} \& \bibinfo{author}{Miguel Palomino}
  (\bibinfo{year}{2009}): \emph{\bibinfo{title}{Non-strongly Stable Orders Also
  Define Interesting Simulation Relations}}.
\newblock In {\sl
  \bibinfo{booktitle}{CALCO'09}}, {\sl \bibinfo{series}{LNCS}} \bibinfo{volume}{5728}, \bibinfo{publisher}{Springer}, pp.
  \bibinfo{pages}{221--235}, \doi{10.1007/978-3-642-03741-2\_16}.

\bibitemdeclare{inproceedings}{FabregasEtAl10-sos}
\bibitem{FabregasEtAl10-sos}
\bibinfo{author}{Ignacio F{\'a}bregas}, \bibinfo{author}{David
  de~Frutos-Escrig} \& \bibinfo{author}{Miguel Palomino}
  (\bibinfo{year}{2010}): \emph{\bibinfo{title}{Equational Characterization of
  Covariant-Contravariant Simulation and Conformance Simulation Semantics}}.
\newblock In {\sl \bibinfo{booktitle}{SOS'10}}, {\sl
  \bibinfo{series}{EPTCS}}~\bibinfo{volume}{32}, pp. \bibinfo{pages}{1--14},
  \doi{10.4204/EPTCS.32.1}.

\bibitemdeclare{inproceedings}{FabregasEtAl10-logics}
\bibitem{FabregasEtAl10-logics}
\bibinfo{author}{Ignacio F{\'a}bregas}, \bibinfo{author}{David
  de~Frutos-Escrig} \& \bibinfo{author}{Miguel Palomino}
  (\bibinfo{year}{2010}): \emph{\bibinfo{title}{Logics for Contravariant
  Simulations}}.
\newblock In {\sl \bibinfo{booktitle}{FORTE-FMOODS 2010}}, {\sl \bibinfo{series}{LNCS}} \bibinfo{volume}{6117}, \bibinfo{publisher}{Springer}, pp.
  \bibinfo{pages}{224--231}, \doi{10.1007/978-3-642-13464-7\_18}.

\bibitemdeclare{incollection}{VanGlabbeek01}
\bibitem{VanGlabbeek01}
\bibinfo{author}{R.~J. van Glabbeek} (\bibinfo{year}{2001}):
  \emph{\bibinfo{title}{The linear time-branching time spectrum {I}: The
  semantics of concrete, sequential processes}}.
\newblock In \bibinfo{editor}{J.~A. Bergstra}, \bibinfo{editor}{A.~Ponse} \&
  \bibinfo{editor}{S.~A. Smolka}, editors: {\sl \bibinfo{booktitle}{Handbook of
  process algebra}}, \bibinfo{publisher}{North-Holland}, pp.
  \bibinfo{pages}{3--99}.

\bibitemdeclare{inproceedings}{Larsen89}
\bibitem{Larsen89}
\bibinfo{author}{Kim~Guldstrand Larsen} (\bibinfo{year}{1989}):
  \emph{\bibinfo{title}{Modal Specifications}}.
\newblock In {\sl
  \bibinfo{booktitle}{Automatic Verification Methods for Finite State
  Systems}}, {\sl \bibinfo{series}{LNCS}}
  \bibinfo{volume}{407}, \bibinfo{publisher}{Springer}, pp.
  \bibinfo{pages}{232--246}, \doi{10.1007/3-540-52148-8\_19}.

\bibitemdeclare{inproceedings}{LarsenT88}
\bibitem{LarsenT88}
\bibinfo{author}{Kim~Guldstrand Larsen} \& \bibinfo{author}{Bent Thomsen}
  (\bibinfo{year}{1988}): \emph{\bibinfo{title}{A Modal Process Logic}}.
\newblock In: {\sl \bibinfo{booktitle}{LICS 1988}}, \bibinfo{publisher}{IEEE
  Computer Society}, pp. \bibinfo{pages}{203--210},
  \doi{10.1109/LICS.1988.5119}.

\bibitemdeclare{inproceedings}{Vogler09}
\bibitem{Vogler09}
\bibinfo{author}{Gerald L{\"u}ttgen} \& \bibinfo{author}{Walter Vogler}
  (\bibinfo{year}{2009}): \emph{\bibinfo{title}{Safe Reasoning with Logic
  LTS}}.
\newblock In {\sl \bibinfo{booktitle}{SOFSEM 2009}}, {\sl
  \bibinfo{series}{LNCS}} \bibinfo{volume}{5404},
  \bibinfo{publisher}{Springer}, pp. \bibinfo{pages}{376--387},
  \doi{10.1007/978-3-540-95891-8\_35}.

\bibitemdeclare{article}{Vogler10}
\bibitem{Vogler10}
\bibinfo{author}{Gerald L{\"u}ttgen} \& \bibinfo{author}{Walter Vogler}
  (\bibinfo{year}{2010}): \emph{\bibinfo{title}{Ready simulation for
  concurrency: It's logical!}}
\newblock {\sl \bibinfo{journal}{Inf. Comput.}}
  \bibinfo{volume}{208}(\bibinfo{number}{7}), pp. \bibinfo{pages}{845--867},
  \doi{10.1016/j.ic.2010.02.001}.

\bibitemdeclare{article}{Lynch88}
\bibitem{Lynch88}
\bibinfo{author}{Nancy Lynch} (\bibinfo{year}{1988}):
  \emph{\bibinfo{title}{{I/O} Automata: A model for discrete event systems}}.
\newblock {\sl \bibinfo{journal}{In 22nd Annual Conferenc  e on Information
  Sciences and Systems}} , pp. \bibinfo{pages}{29--38}.\\ 
\newblock{\url{http://groups.csail.mit.edu/tds/papers/Lynch/MIT-LCS-TM-351.pdf}}

\bibitemdeclare{book}{Milner89}
\bibitem{Milner89}
\bibinfo{author}{R.~Milner} (\bibinfo{year}{1989}):
  \emph{\bibinfo{title}{Communication and Concurrency}}.
\newblock \bibinfo{publisher}{Prentice Hall}.

\bibitemdeclare{inproceedings}{Park81}
\bibitem{Park81}
\bibinfo{author}{David Park} (\bibinfo{year}{1981}):
  \emph{\bibinfo{title}{Concurrency and Automata on Infinite Sequences}}.
\newblock In {\sl
  \bibinfo{booktitle}{Theoretical Computer Science, 5th GI-Conference}}, {\sl
  \bibinfo{series}{LNCS}} \bibinfo{volume}{104},
  \bibinfo{publisher}{Springer}, pp. \bibinfo{pages}{167--183},
  \doi{10.1007/BFb0017309}.

\end{thebibliography}

\end{document}